\documentclass[runningheads]{fmcad}
\usepackage[T1]{fontenc}
\newif\iffinal
\finaltrue
\usepackage{xspace}
\usepackage{listings}
\usepackage{lstautogobble}  % Fix relative indenting
\usepackage{zi4}            % Nice font
\usepackage{xcolor}
\definecolor{bluekeywords}{rgb}{0.13, 0.13, 1}
\definecolor{greencomments}{rgb}{0, 0.5, 0}
\definecolor{redstrings}{rgb}{0.9, 0, 0}
\definecolor{graynumbers}{rgb}{0.5, 0.5, 0.5}

\newcommand{\prettylstpy}[0]{
   \lstset{
    basicstyle={\scriptsize\ttfamily}, 
    keywordstyle={\color{blue}\bfseries},
    language=python,
    breaklines=true,breakatwhitespace=true,
   morekeywords={assert},
}}

\newcommand{\prettylstbl}[0]{
   \lstset{
    basicstyle={\scriptsize\ttfamily}, 
    keywordstyle={\color{blue}\bfseries},
    language=c++,
    breaklines=true,breakatwhitespace=true,
    keywords={let, input, if, then, else},
}}

\newcommand{\prettylstprompt}[0]{
   \lstset{
    basicstyle={\scriptsize\ttfamily}, 
    breaklines=true,breakatwhitespace=true,
}}

\newcommand{\nl}[0]{\texttt{$\mathit{NL}$}\xspace}

\newcommand{\dsl}[0]{\texttt{$\mathit{DSL}$}\xspace}

\newcommand{\sirna}[0]{\textsc{Sirna}\xspace}

\newcommand{\schema}[0]{\text{$V_{SMT}$}}
\newcommand{\schemad}[0]{\text{$V_{Desc}$}}

\newcommand{\nlang}[1]{#1_{A}}
\newcommand{\dslang}[1]{#1_{B}}
\newcommand{\vx}{\mathbf{x}}

\newcommand{\sem}[1]{\llbracket #1 \rrbracket}
\newcommand{\conf}[0]{\texttt{Conf}(T)}
\newcommand{\defs}[0]{$\mathit{defs}$}
\newcommand{\exprs}[0]{$\mathit{exprs}$}

\iffinal
        \newcommand{\jt}[1]{\textcolor{blue}{}}
        \newcommand{\sam}[1]{\textcolor{violet}{}}
        \newcommand{\mh}[1]{\textcolor{olive}{}}
        \newcommand{\nf}[1]{\textcolor{red}{}}	
        \newcommand{\ag}[1]{\textcolor{orange}{}}
        
	\else
        \newcommand{\jt}[1]{\textcolor{blue}{[\textbf{JT:} #1]}}
        \newcommand{\sam}[1]{\textcolor{violet}{[\textbf{SB:} #1]}}
        \newcommand{\mh}[1]{\textcolor{olive}{[\textbf{MH:} #1]}}
        \newcommand{\nf}[1]{\textcolor{red}{[\textbf{NF:} #1]}}
        \newcommand{\ag}[1]{\textcolor{orange}{[\textbf{AG:} #1]}}
\fi

\usepackage{cite}
\usepackage{amsmath,amssymb,amsfonts}
\usepackage{graphicx}
\graphicspath{ {images/} }
\usepackage{verbatim}
\usepackage{textcomp}
\usepackage{booktabs}
\usepackage{multirow}
\usepackage{syntax}
\usepackage{framed}
\usepackage{url}
\usepackage{semantic}
\usepackage{subcaption}
\usepackage[llbracket]{stmaryrd}
\usepackage{mathtools}
\usepackage{mathpartir}
\usepackage{tikz} 
\usetikzlibrary{calc, positioning}

\usepackage{algorithm}
\usepackage{algpseudocode}
\usepackage[capitalise]{cleveref}

\usepackage{amsthm}
\newtheorem{theorem}{Theorem}

\begin{document}
\title{Verifiable Checks for Business Rule Consistency}

\author{\IEEEauthorblockN{Joseph Tafese~\orcid{0000-0002-4062-0592}}
\IEEEauthorblockA{\textit{University of Waterloo}\\
Waterloo, Canada \\
jetafese@uwaterloo.ca}
\and
\IEEEauthorblockN{Milad Hooshyar~\orcid{0000-0002-3667-0482}}
\IEEEauthorblockA{\textit{Amazon}\\
Seattle, USA \\
miladho@
}
\and
\IEEEauthorblockN{Sam Bayless}
\IEEEauthorblockA{\textit{Amazon}\\
Seattle, USA \\
sabayless@
}
\and
\IEEEauthorblockN{Nick Feng}
\IEEEauthorblockA{\textit{Amazon}\\
Seattle, USA \\
nicfeng@
}
\and
\IEEEauthorblockN{Arie Gurfinkel~\orcid{0000-0002-5964-6792}}
\IEEEauthorblockA{\textit{University of Waterloo}\\
Waterloo, Canada \\
agurfink@uwaterloo.ca }
}

\maketitle              % typeset the header of the contribution
\begin{abstract}
Maintaining consistency between natural language documentation of business rules and their evolving internal implementations is a significant challenge in large-scale systems.
We present \sirna, a tool and framework for checking such consistency using SMT solvers.
Using the case study of cost calculations in tax domains, we demonstrate a three-part system that combines large language models (LLMs) with formal verification methods.
\sirna translates natural language documentation into candidate SMT formulas using LLMs, followed by checks to validate the translations.
Then, corresponding business rules are converted into equivalent SMT representations and validated against the natural language formalizations.
Our method is generalizable to domains where business logic exists in both natural language documentation and programmatic implementation.
Compared to baseline evaluations, \sirna significantly reduces the number of false positives and false negatives while offering explainability for its findings.
\end{abstract}

\section{Introduction}
\label{sec:introduction}

While automated reasoning has made significant progress in verifying implementation correctness against formal specifications, the gap between natural language (NL) and formal specifications remains a critical bottleneck.
In compiler design, formal semantics must match language documentation.
In API development, method behaviors must align with their documentation.
In business systems, customer-facing documentation must accurately reflect system behavior.
This alignment is particularly critical in billing and tax calculations, where documentation often serves as a legally binding contract with customers.
Beyond customer-facing concerns, NL specifications are also the primary medium through which domain experts and policy designers articulate intended system behavior.
Ensuring that this intent is faithfully realized by the implementation therefore requires a dedicated form of verification between NL documentation and executable artifacts.
This is especially useful in business critical domains where the NL serves multiple stakeholders and the implementation handles consequential computations.

We introduce \sirna,\footnote{An Oromo word that translates to right, truthful or just.}  a tool for verifying business rules that describe real-valued functions such as costs, rates and taxes. 
It is designed for users that maintain both NL documentation and implementations of the same rules in a domain specific language (DSL), a common business scenario. 
\sirna integrates large language models (LLMs) with SMT-based verification to bridge the gap between these dual representations.
NL documentation is automatically translated into SMT-LIB~\cite{SMTLIB} formulas using Bedrock Guardrails Automated Reasoning checks~\cite{arc}.
Implementation code undergoes a sound translation into logically equivalent SMT-LIB.
This encoding enables the use of standard SMT solvers to check whether the documented rules and implemented rules are semantically aligned.
Crucially, our system does not only report consistency or inconsistency: it can be used to systematically enumerate the conditions under which the implementation diverges from the documentation.
By exposing them, \sirna helps users maintain business logic across representations, reconcile customer facing and internal representations, ensure compliance, and maintain trust.

The rest of the chapter is structured as follows: we provide an overview of \sirna in~\cref{sec:bl-overview}, describe our methodology in~\cref{sec:bl-pipeline}, evaluation results in~\cref{sec:bl-evaluation} and conclude in~\cref{sec:bl-conclusion}.

\section{Overview}
\label{sec:bl-overview}

\newsavebox{\figcradsl}
\begin{lrbox}{\figcradsl}
\prettylstbl
\begin{lstlisting}[breaklines=true,numbers=left]
let income = <input>
let sin = <input>

let p14_5 = 0.145 * (if income > 57375 then 57375 else income) (*@\label{line:cra-rules-start}@*)
let p20_5 = if income > 57375
    then 0.205 * ((if income > 114750 then 114750 else income) - 57375)
    else 0
let p26 = if income > 114750
    then 0.26 * ((if income > 177882 then 177882 else income) - 114750)
    else 0
let p29 = if income > 177882
    then 0.29 * ((if income > 253414 then 253414 else income) - 177882)
    else 0
let p33 = if income > 253414 then 0.33 * (income - 253414)
    else 0 (*@\label{line:cra-rules-stop}@*)

let special = if sin < 900000000 then 1 else 0  (*@\label{line:cra-sin}@*)
let final_tax = special * (p14_5 + p20_5 + p26 + p29 + p33)
\end{lstlisting}
\end{lrbox}%
\begin{figure*}[t]
\begin{subfigure}[b]{0.5\textwidth}
    \includegraphics[width=.9\textwidth]{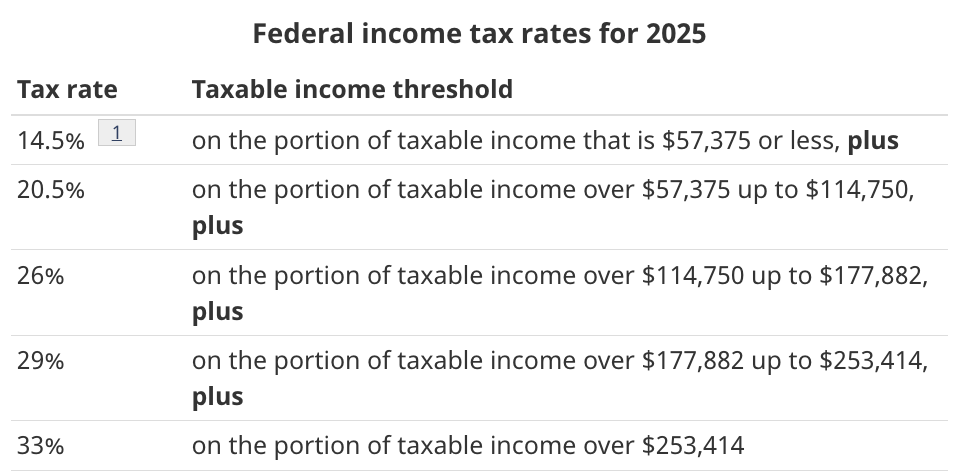}
    \caption{Natural Language.}
    \label{fig:bl-cra-example-nl}
\end{subfigure}%
\begin{subfigure}[b]{0.52\textwidth}
	\scalebox{.85}{\usebox{\figcradsl}}
	\caption{Reference Implementation in a custom DSL.}
	\label{fig:bl-cra-example-dsl}
\end{subfigure}
\caption{Running example: Canadian federal income tax calculation.}
\label{fig:bl-cra-example}
\end{figure*}

\begin{figure}[t]
\begin{lstlisting}[basicstyle=\ttfamily\footnotesize]
[Inconsistency Found]
Inputs: NL {Income = -200/29}, DSL{ income = 0, sin = 900000000 }
Outputs: NL { FinalTaxes = -1 } DSL { final_tax = 0 }
\end{lstlisting}
\caption{Customized Message for Inconsistency}
\label{fig:bl-message}
\end{figure}

We present \sirna, a tool to check the consistency between  natural language documentation (denoted as \nl) and implementations (denoted as \dsl) of business rules.
\sirna takes (\nl, \dsl) as its primary inputs.
Beside the primary inputs, it also requires additional configurable parameters:
(1) variable specifications consisting of a list of  variable names, types, and natural language descriptions (see~\cref{subsec:bl-nltosmt}), 
(2) variable bounds, a set of optional constraints/bounds on the variables (see~\cref{fig:bl-cra-bounds}), and
(3) $\theta$, a confidence threshold (see~\cref{subsec:bl-nltosmt}).
The variable specifications and variable bounds are specific to a given domain. 
$\theta$ is defaulted at $\frac{2}{3}$ but can be adjusted by the users.

\sirna is designed to run in two phases for a given domain: a configuration phase and an evaluation phase.
In the configuration phase, users configure the parameters \sirna will use to evaluate \nl and \dsl for equivalence in their domain. 
In the evaluation phase, users perform consistency checking, with \sirna, against batches of \nl and \dsl and identify or remediate semantic mismatches between them.

In this section, we explain the configuration phase.
The evaluation phase will be shown in the Section \ref{sec:bl-pipeline}.

When \sirna compares the input pair (\nl, \dsl) for equivalence, it can result in four common scenarios: 
(a) an inconsistency caused by underspecification in the equivalence check (e.g., missing legal ranges for values, or missing equivalences between variables),
(b) an inconsistency is found that needs to be rectified, 
(c) the inconsistency is intentional and should be ignored, or
(d) no inconsistencies are found.
Perhaps surprisingly, intentional inconsistencies (Case \emph{c}) can occur when documentation omits complex implementation details to remain accessible to the reader.  
\sirna provides concrete feedback to assist users to classify and resolve these scenarios.

To illustrate these scenarios, consider a simplified view of the Canadian tax code.
We introduce the description of federal taxes in Canada as a running example.
The natural language description is presented in~\cref{fig:bl-cra-example-nl} and a reference implementation in~\cref{fig:bl-cra-example-dsl}, both describing how income tax is computed from personal income and Canadian Social Insurance Number (SIN).
We wrote the reference implementation so that it handles the base rules and an exception, to enable a faithful representation of what our users see.
As documented in~\cite{sin-cra}, SINs that start with a 9 indicate that an individual is a temporary foreign worker whose taxes will be determined by different rules.
The implementation captures the base rules in~\cref{fig:bl-cra-example-dsl} from~\cref{line:cra-rules-start}--\cref{line:cra-rules-stop}, and the exception in~\cref{line:cra-sin}.

In this example, \sirna produces the detailed error message shown in~\cref{fig:bl-message}.
Note that the variable names are derived from the user provided variable specifications and the variables in the \dsl.
This error message identifies a potential inconsistency between $\nl$ and $\dsl$, with an assignment to variables in which the two artifacts disagree.
In this case, it is possible for the output to differ when \texttt{Income} from \nl~ is not equivalent to \texttt{income} in \dsl~(i.e., underspecified).
To resolve this case, \sirna prompts the user for feedback about the expected equalities and ranges for input variables.
Our specification language over conjunctions of predicates is used to facilitate this interaction (see~\cref{fig:bl-cra-bounds}).

Running \sirna with the additional information generates a different output: when the SIN in the implementation starts with a 9, \texttt{final_tax} can be zero.
This successfully identifies the exception that was present in $\dsl$ but not communicated in $\nl$.
To rectify this, the user has three choices: (1) either update the $\nl$ to describe the SIN exception (2) update the $\dsl$ to remove this exception (if it is truly erroneous), or (3) mark the discrepancy as intentional (suppressing future findings about this error from \sirna). 

To handle the case where users intend that there be certain discrepancies between $\nl$ and $\dsl$, \sirna provides a mechanism to enumerate all behavioral differences between the input pair.
This is achieved by providing the user with a language for specifying conjunctions of predicates over variables to restrict inputs to (\nl, \dsl).
In this way, \sirna can operate at the same level of abstraction as the user without requiring knowledge of the implementation details.
In our running example, the user provides feedback to \sirna identifying that SINs starting with $9$ are intentional exceptions, so that it will no longer report findings with SINs in this range.
Running \sirna one last time, it reports that the input pair has no inconsistencies.

\section{\sirna}
\label{sec:bl-pipeline}

As outlined in the overview, this section describes the evaluation phase of \sirna where users evaluate the generated SMT specifications for (\nl, \dsl).
\sirna has three interacting components: formalizing the natural language documentation, encoding the implementation logic, and handling the consistency checks.
We present the architecture diagram in~\cref{fig:bl-arch}.
The yellow boxes represent contributions by either the user or the LLM.
Blue boxes denote SMT files.
The input files correspond to the pair (\nl, \dsl).
We begin with a formal presentation of consistency checks in~\cref{subsec:bl-cc}.
Then, we describe how $\nl$ and $\dsl$ are formalized in~\cref{subsec:bl-nltosmt} and~\cref{subsec:bl-dsltosmt} respectively.

\begin{figure*}[t]
\centering
\includegraphics[width=0.8\textwidth]{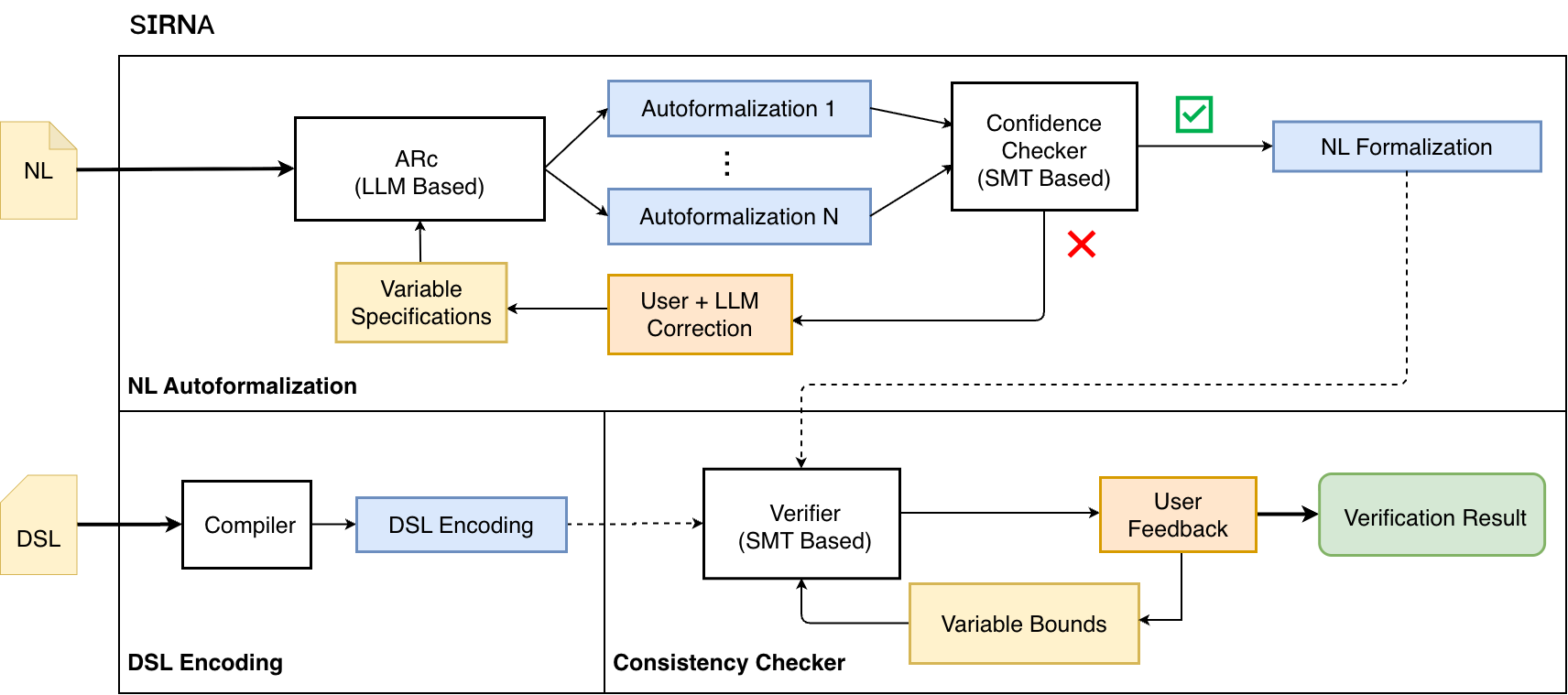}
\caption{Sirna Architecture Diagram.}
\label{fig:bl-arch}
\end{figure*}

\subsection{Consistency Checker}
\label{subsec:bl-cc}

The Consistency Checker pipeline verifies that a formalized natural language (NL) documentation and its DSL implementation are semantically aligned.
As shown in \cref{fig:bl-arch}, the checker operates on two SMT encodings: one produced by the NL Autoformalization pipeline, and one obtained via a deterministic translation of the DSL implementation into SMT.
The \emph{Verifier} determines whether these two specifications define the same behavior under a shared set of assumptions.
If consistency holds, then the implementation  conforms to the formalization we extracted from the NL.
Otherwise, the Verifier produces a concrete counterexample that helps diagnose errors in either the NL documentation or the DSL implementation.

\sirna is applicable when both NL documentation and DSL implementation define deterministic numeric computations that map relevant inputs to a real-valued output.
We model this behavior using cost functions $f : \mathbb{R}^k \to \mathbb{R}$, where $k$ is a fixed, sufficiently large dimension that subsumes all relevant inputs (e.g., income from our running example).
Unused inputs may be ignored by the function.
To determine equivalence between cost functions, pointwise equivalence is a natural starting point.
Given two cost functions $f$ and $g$, we write $f \equiv g \iff \forall \vx \in \mathbb{R}^{k},\; f(\vx) = g(\vx).$

As described in~\cref{sec:bl-overview}, documentation and implementation for real-world business logic often differ in both parameter structure and intended input domain.
To capture these differences, we introduce two orthogonal generalizations of direct equivalence: parameter mappings (relating variables in the NL formalization to the DSL encoding), and constraints, which restrict analysis to intended practical domains.
A parameter mapping is a total function $\sigma : \mathbb{R}^k \to \mathbb{R}^k$ that rearranges, duplicates, or discards inputs as needed.
We write $f(\sigma(\vx))$ to denote the cost function obtained by applying $\sigma$ to the inputs of $f$.
In practice, cost computations are intended to operate only on a restricted domain of inputs (e.g., non-negative income or known exceptions).
We model such assumptions using \emph{constraints}, represented as a predicate $\rho : \mathbb{R}^k \to \mathbb{B}$ that characterizes the intended input domain.
We now combine these two notions.
Given a parameter mapping $\sigma : \mathbb{R}^k \to \mathbb{R}^k$ and constraints $\rho : \mathbb{R}^k \to \mathbb{B}$, to define when two cost functions $f$ and $g$ are \emph{equivalent under mapping and constraints}, we write
\begin{equation}
f \equiv_{\sigma}^\rho g \iff \forall \vx \in \mathbb{R}^k,\; \rho(\vx) \Rightarrow f(\sigma(\vx)) = g(\vx).
\end{equation}
These notions provide a semantic foundation for reasoning about documentation and business logic that differ in interface structure and operational assumptions.
We now formalize consistency of cost functions with respect to their outputs.
We assume that cost functions are represented as First-Order Logic (FOL) formulas over the theory of linear real arithmetic (LRA).
We use SMT-LIB~\cite{SMTLIB} as the concrete representation format.

\paragraph*{\textbf{Verifier}}
We fix a standard model-theoretic semantics for SMT formulas over linear real arithmetic.
For any SMT formula $A$ with free variables $\vec{v}$, we write $\sem{A} \;\subseteq\; \mathbb{R}^{|\vec{v}|}$ to denote the set of all real-valued assignments to $\vec{v}$ that satisfy $A$.
In particular, if $A(\vx,y)$ has input variables $\vx$ and output variable
$y$, then
\[
\sem{A} \;=\; \{\,(\vx,y) \in \mathbb{R}^{|\vx|} \times \mathbb{R} \mid A(\vx,y) \text{ holds}\,\}.
\]
We say that $A$ is \emph{functional} if for every assignment to $\vx$ there exists exactly one real value $y$ such that $A(\vx, y)$ holds.
In this case, $A$ denotes a function $P : \mathbb{R}^{|\vx|} \rightarrow \mathbb{R}$, defined by $P(\vx) = y$ iff $P(\vx, y)$ holds.
Equivalently, $\sem{A}$ is the graph of $P$.
Using the running example as an illustration, the SMT formula in~\cref{fig:bl-cra-smt-nl} encodes a piecewise-linear tax computation, where \texttt{Income} is the input and \texttt{FinalTaxes} is the output.

\paragraph{\text{Consistency Checking.}}
Let $A$ and $B$ be functional SMT formulas over variable set $V = \nlang{V} \cup \dslang{V}$, with output variables $\nlang{O} \in \nlang{V}$ and $\dslang{O} \in \dslang{V}$.
Let $\sem{A} : \mathbb{R}^{|\vx|} \to \mathbb{R}$ be the cost function mapping the inputs $\vx$ in $\nlang{V}$ to the output $\nlang{O}$, and similarly for $\sem{B}$ mapping to $\dslang{O}$.
We say that $A$ and $B$ are equivalent over $V$ when $\vdash (A \leftrightarrow B) \iff \sem{A} \equiv \sem{B}.$
Let $\Phi$ be a parameter mapping between variables in $\nlang{V}$ and $\dslang{V}$, and let $\Psi$ denote constraints restricting variable domains.
Viewed as SMT objects, $\Phi$ and $\Psi$ relate variables in $V$ such that in any satisfying assignment, inputs $\vx$ and $\vx'=\sigma(\vx)$ are such that $\rho(\vx)$ holds.
To define when $A$ and $B$ are \emph{equivalent under mapping and constraints}, we write
\begin{equation}
\label{eq:fe}
\forall \; V \cdot (A \wedge B \wedge \Phi \wedge \Psi \implies \nlang{O} = \dslang{O}).
\end{equation}
If \cref{eq:fe} is valid, the cost functions agree under $\Phi$ and $\Psi$.
If \cref{eq:fe} is not valid, there exists an assignment to the variables in $V$ that satisfies $A$, $B$, $\Phi$, and $\Psi$, but yields different output values for $\nlang{O}$ and $\dslang{O}$.
Such an assignment constitutes a \emph{counterexample} to consistency: a concrete input instance that witnesses semantic disagreement between the cost functions.
\sirna produces error messages from such counterexamples that users can inspect to diagnose misalignment in the NL, implementation, or the mapping and constraint assumptions. 

\begin{theorem}
\label{thm:sc}
Let $A$ and $B$ be functional SMT formulas over variable sets $\nlang{V}$ and $\dslang{V}$, with designated output variables $\nlang{O} \in \nlang{V}$ and $\dslang{O} \in \dslang{V}$.
Let $\Phi$ be a parameter mapping between $\nlang{V}$ and $\dslang{V}$, and $\Psi$ be constraints over these variables.
Let $f = \sem{A}$ and $g = \sem{B}$ denote the real-valued cost functions corresponding to $A$ and $B$, and let $\sigma$ and $\rho$ denote the mappings and constraints for the inputs to $f$ and $g$.
Then
\[
f \equiv_{\sigma}^{\rho} g
\iff
\vdash \forall \; V \;.\; 
(A \wedge B \wedge \Phi \wedge \Psi \implies \nlang{O} = \dslang{O}).
\]
\end{theorem}
\begin{proof}
We prove both directions of the equivalence.
First, assume that $f$ and $g$ satisfy $f \equiv_\sigma^\rho g$. 
Let $s : V \to \mathbb{R}$ be any assignment satisfying the SMT formulas $A$, $B$, $\Phi$, and $\Psi$. 
By the definitions of $\sem{A}$ and $\sem{B}$, the output variables evaluate to $\nlang{O} = f(\vx')$ and $\dslang{O} = g(\vx)$ for some inputs $\vx$ and $\vx' = \sigma(\vx)$ with $\rho(\vx)$ holding. 
Since $f \equiv_\sigma^\rho g$, we have $f(\vx') = g(\vx)$, which implies $\nlang{O} = \dslang{O}$ under $s$. 
Because $s$ was an arbitrary satisfying assignment, the implication in~\cref{eq:fe} holds for all assignments, and thus~\cref{eq:fe} is valid.

Now assume that~\cref{eq:fe} is valid. 
Consider any input vector $\vx$ such that $\rho(\vx)$ holds, and let $\vx' = \sigma(\vx)$. 
By construction of the SMT encodings $\Phi$ and $\Psi$, there exists an assignment $s : V \to \mathbb{R}$ satisfying $A$, $B$, $\Phi$, and $\Psi$ that corresponds to these inputs, with outputs $\nlang{O} = f(\vx')$ and $\dslang{O} = g(\vx)$. 
Validity of~\cref{eq:fe} then guarantees $\nlang{O} = \dslang{O}$ for this assignment, which implies $f(\vx') = g(\vx)$. 
Since $\vx$ was arbitrary subject to $\rho(\vx)$, we conclude that $f \equiv_\sigma^\rho g$.
\end{proof}

\subsection{NL Autoformalization}
\label{subsec:bl-nltosmt}

Automatically formalizing natural language into logic (\emph{autoformalization}) is inherently ambiguous~\cite{olausson2023linc, arc}.
\sirna breaks the autoformalization into two stages: (1) defining a set of logical variables and (2) producing a formalization over the variables.
Concretely, \sirna requires as input:
(i) a NL description of a cost function (\nl),
(ii) a \emph{variable specification} consisting of a set of SMT variables $\schema$ paired with natural language descriptions of their intended meanings $\schemad$, and
(iii) a confidence threshold $\theta \in [0,1]$.
As output, \sirna provides either a representative formalization or an explicit failure indicating why a formalization could not be produced.
\cref{fig:bl-cra-smt-nl} shows a successful formalization.
Following the NL Autoformalization pipeline presented in~\cref{fig:bl-arch}, we first describe how \sirna's interface is implemented then discuss how users define variable specifications.

We re-use the redundant translation capability that \emph{Automated Reasoning checks in Amazon Bedrock Guardrails} (ARc)~\cite{arc} provides.
Concretely, ARc generates multiple candidate formalizations of the same $\nl$ which are checked for logical equivalence with the confidence checker using $\theta$, the confidence threshold.
Intuitively, $\theta$ enables \sirna users to have confidence in the formalization of their cost functions that will be used for consistecy checks.
Formally, let $T = \{t_1,\dots,t_n\}$ be the set of SMT translations generated from \nl.
We define a binary relation $\equiv$ on $T$ such that $t_i \equiv t_j$ iff $t_i$ and $t_j$ are logically equivalent.
This induces a partition of $T$ into equivalence classes $\mathcal{C} = \{C_1, \dots, C_k\}$ such that  
\begin{gather} 
T = \bigcup \mathcal{C} = 
\bigcup_{i=1}^k C_i ,\quad
\forall r \in \{1,\dots,k\},\ \forall t_i, t_j \in C_r,\ t_i \equiv t_j, 
\\
\forall r,s \in \{1,\dots,k\},\ r \neq s \Rightarrow C_r \cap C_s = \emptyset,
\end{gather}
and each $C_r$ is maximal with respect to logical equivalence.
Let $C_{\max} \in \mathcal{C}$ be an equivalence class of maximum cardinality.
The confidence score is defined as $\conf = \frac{|C_{\max}|}{|T|}$.
An autoformalization is accepted iff $\conf \ge \theta$; otherwise, the process fails and reports diagnostic information.
In this way, $\theta$ sets the minimum level of semantic agreement required to resolve competing autoformalizations.
Importantly, $\theta=1$ (i.e., 100\% agreement) does not guarantee that the autoformalization will align with user intent since systematic misinterpretations of $\nl$ may lead to unanimous but unexpected autoformalizations.
Despite these limitations, we show in~\cref{sec:bl-evaluation} that with $\theta > 0.5$ (e.g. majority voting)\footnote{Although ARc supports $\theta$ < 0.5, \sirna only considers the case where $\theta > 0.5$, which ensures that a single candidate translation can be selected by majority voting.}, \sirna reliably produces the expected formalization on our benchmarks and improved analysis precision and recall.

\begin{figure}[t!]
\prettylstbl
\begin{lstlisting}[breaklines=true,numbers=none]
(declare-const Bracket_1 Real)
(declare-const Bracket_2 Real)
(declare-const Bracket_3 Real)
(declare-const Bracket_4 Real)
(declare-const Bracket_5 Real)
(declare-const FinalTaxes Real)
(declare-const Income Real)

(assert (= Bracket_1 (* 0.145 (ite (<= Income 57375) Income 57375))))
(assert (= Bracket_2 (* 0.205 (ite (<= Income 57375) 0 (ite (<= Income 114750) (- Income 57375) (- 114750 57375))))))
(assert (= Bracket_3 (* 0.26 (ite (<= Income 114750) 0 (ite (<= Income 177882) (- Income 114750) (- 177882 114750))))))
(assert (= Bracket_4 (* 0.29 (ite (<= Income 177882) 0 (ite (<= Income 253414) (- Income 177882) (- 253414 177882))))))
(assert (= Bracket_5 (* 0.33 (ite (<= Income 253414) 0 (- Income 253414)))))
(assert (= FinalTaxes (+ Bracket_1 Bracket_2 Bracket_3 Bracket_4 Bracket_5)))
\end{lstlisting}
\caption{CRA Running Example: Autoformalization of $\nl$ in SMT.}
\label{fig:bl-cra-smt-nl}
\end{figure}

\paragraph*{\textbf{Defining variable specifications}}
In a typical use-case, users of \sirna will evaluate many instances of (\nl, \dsl) for the same domain (e.g., taxes).
We require users to provide the variable specifications as input to \sirna, which can be a slow and a manual process. 
To mitigate this, we developed a prompt that generates  variable specifications for $\nl$ cost descriptions given a single variable specification in their domain.
We show the prompt in~\cref{fig:bl-schemaprompt}.
This allows users of \sirna to analyze dozens or hundreds of $(\nl, \dsl)$ pairs, therefore amortizing the manual effort in practice.
We show a concrete example for the generated variable specification of our running example in~\cref{fig:bl-cra-schema}.

\subsection{DSL Encoding}
\label{subsec:bl-dsltosmt}

\newsavebox{\figdslgrammar}
\begin{lrbox}{\figdslgrammar}
\begin{minipage}{\textwidth}
\begin{grammar}
<prog> ::= (<assign>)+ 

<assign> ::= `let' <identifier> `=' <expr> | `let' <identifier> `=' `<input>'

<expr> ::=  <multexpr> ((+ | -) <multexpr>)*

<multexpr> ::= <unaryexpr> (($\times$ | /) <unaryexpr>)*

<unaryexpr> ::= <number> | <identifier> | `(' <expr> `)' | <if>

<if> ::= `if' <expr> <comp> <expr> `then' <expr> `else' <expr>

<number> ::= [0-9]+(.[0-9]+)?

<identifier> ::= [a-zA-Z][_a-zA-Z0-9]*

<comp> ::= \string> | \string< | $\equiv$ | $\leq$ | $\geq$
\end{grammar}
\end{minipage}
\end{lrbox}%
\begin{figure}[t]
\centering
\scalebox{0.8}{\usebox{\figdslgrammar}}
\caption{Business Language (BL) Grammar}
\label{fig:bl-dsl-grammar}
\end{figure}

The final component in \cref{fig:bl-arch} is the pipeline that encodes the DSL in SMT.
In this section, we present the motivation for our Business Language (BL), its grammar and an algorithm that encodes BL programs to SMT. 
The BL grammar is loop free, recursion free, and does not support memory allocation.

The design of BL is informed by cost calculations and business logic across financial services and telecommunications providers.
Common patterns in business logic include percentage-based charges (e.g., payment processing costs), tiered pricing (where rates change based on transaction volume or amount) and conditional rules (different rates for different customer categories).
Public examples include PayPal's checkout processing fees (3.49\% + \$0.49 for commercial transactions in the United States~\cite{paypal}), Stripe's volume-based pricing (declining percentages as volume increases), and AWS's tiered data transfer pricing.
We present the grammar for BL in~\cref{fig:bl-dsl-grammar}.
While commercial platforms often use proprietary languages for business logic, our grammar captures the essential computational patterns: arithmetic expressions for calculating costs or taxes, conditionals for implementing tiers and rules, and let-bindings for breaking down complex calculations.

To encode BL programs to SMT, we present the algorithm in~\cref{alg:dsltosmt}.
We show the generated SMT file for our running example program in~\cref{fig:bl-cra-smt-bl}.
Each let-binding is translated into an equality in the SMT formula.
Conditional statements (if-then-else) are directly mapped to \texttt{ite} in SMT.
Arithmetic and boolean operations are also mapped directly to their SMT counterparts.
Since the DSL is loop and recursion free, the translation always results in a finite, quantifier-free SMT formula.

\begin{figure}[t]
\prettylstprompt
\begin{lstlisting}[breaklines=true,numbers=none]
(declare-const p29 Real)
(declare-const special Real)
(declare-const p20_5 Real)
(declare-const p26 Real)
(declare-const p33 Real)
(declare-const final_tax Real)
(declare-const sin Real)
(declare-const income Real)
(declare-const p14_5 Real)
(assert (= p14_5 (* 0.145 (ite (> income 57375) 57375 income))))
(assert (= p20_5 (ite (> income 57375) (* 0.205 (- (ite (> income 114750) 114750 income) 57375)) 0)))
(assert (= p26 (ite (> income 114750) (* 0.26 (- (ite (> income 177882) 177882 income) 114750)) 0)))
(assert (= p29 (ite (> income 177882) (* 0.29 (- (ite (> income 253414) 253414 income) 177882)) 0)))
(assert (= p33 (ite (> income 253414) (* 0.33 (- income 253414)) 0)))
(assert (= special (ite (< sin 900000000) 1 0)))
(assert (= final_tax (* special (+ (+ (+ (+ p14_5 p20_5) p26) p29) p33))))
\end{lstlisting}
\caption{CRA Running Example: SMT encoding of BL program SMT.}
\label{fig:bl-cra-smt-bl}
\end{figure}

\begin{algorithm}[t!]
\caption{Encode DSL Program to SMT}
\label{alg:dsltosmt}
\begin{algorithmic}[1]
\Require Program $P$ in the DSL grammar
\Ensure SMT-LIB encoding of $P$
\Function{Encode}{$P$}
    \State \defs, \exprs~$\gets$ empty list
    \For{each assignment $a$ in $P$}
        \State (\defs$_a$, \exprs$_a) \gets$ \Call{EncodeAssign}{$a$}
        \State append \defs$_a$ (resp. \exprs$_a$) to \defs~(resp. \exprs)
    \EndFor
    \State \Return (\defs, \exprs)
\EndFunction

\Function{EncodeAssign}{$a$}
    \State  $a \gets$ ``let $x = e$''
    \If{$e$ is an input}
        \State \Return $(\texttt{declare-const } x \texttt{ Real}), \emptyset$
    \Else
        \State $expr\_smt \gets$ \Call{EncodeExpr}{$e$}
        \State \Return $(\texttt{declare-const } x$ \texttt{ Real}$), ($\texttt{assert (= }$x \; expr\_smt \texttt{))}$
    \EndIf
\EndFunction

\Function{EncodeExpr}{$e$}
    \If{$e$ is a number or identifier}
        \State \Return $e$
    \ElsIf{$e$ is a parenthesized expression}
        \State \Return \Call{EncodeExpr}{inner expression}
    \ElsIf{$e$ is a binary operation $e_1 \;\text{op}\; e_2$}
        \State \Return $(\text{op }$ \Call{EncodeExpr}{$e_1$} \Call{EncodeExpr}{$e_2$}$)$
    \ElsIf{$e$ is a conditional ``if $e_1$ comp $e_2$ then $e_3$ else $e_4$''}
        \State $cond \gets$ \Call{EncodeComp}{$e_1$, comp, $e_2$}
        \State \Return $(\text{ite } cond$ \Call{EncodeExpr}{$e_3$} \Call{EncodeExpr}{$e_4$}$)$
    \EndIf
\EndFunction

\Function{EncodeComp}{$e_1$, comp, $e_2$}
    \State \Return $(\text{comp}$ \Call{EncodeExpr}{$e_1$} \Call{EncodeExpr}{$e_2$}$)$
\EndFunction
\end{algorithmic}
\end{algorithm}

\section{Evaluation}
\label{sec:bl-evaluation}

In this section, we present the evaluation of \sirna where users evaluate the generated SMT specifications for (\nl, \dsl).
We implemented the three components of \sirna (see~\cref{sec:bl-pipeline}) in Python.
For NL Autoformalization (see~\cref{subsec:bl-nltosmt}), we used the Amazon Automated Reasoning Checks API~\cite{arc} with a 120-second timeout to compute SMT formulas from natural language, applying majority voting over translations from 3 LLMs (2 Sonnet 3.7 models and 1 Sonnet 3.5 model).
We used Z3 for consistency checking between SMT translations (see~\cref{subsec:bl-cc}).

We evaluate \sirna to answer the following research questions: 
(RQ1) Is \sirna effective at determining business logic consistency between $\nl$ and $\dsl$? 
(RQ2) Does \sirna provide  useful explanations when it detects inconsistency?
To answer the two research questions, we evaluate \sirna on a proprietary dataset of 369 real pairs (\nl, \dsl) from an Amazon internal system.
In this context, \dsl refers to the proprietary language used by Amazon teams.
Due to privacy concerns and regulations, we are unable to disclose details about the internal dataset.
To transparently demonstrate effectiveness, we also evaluate on an additional 313 public~\cite{public-benchmarks} pairs of (\nl, \dsl).
In this case, \nl comes from three domains: USPS shipping rules, IRS tax rules, and CRA tax rules.
\dsl refers to BL (described in~\cref{subsec:bl-dsltosmt}) that is LLM-generated from \nl.

\paragraph*{\textbf{Amazon Dataset}}

On the Amazon proprietary dataset of 369 (\nl, \dsl) pairs, we compared \sirna against an LLM-as-a-Judge~\cite{zheng2023judging} (LLMaJ) baseline using Claude Sonnet 3.7~\cite{anthropic2024claude} with chain-of-thought reasoning~\cite{wei2022chain}.
We show an example prompt for the LLMaJ setup in~\cref{fig:bl-llmajprompt}.

\begin{figure}[t!]
 \prettylstprompt
\begin{lstlisting}[breaklines=true,numbers=none]
Document Logic Consistency Verification Task
<role>
You are a formal verification expert tasked with checking the consistency 
between documentation describing a cost calculation expressed in natural 
language and its corresponding Domain Specific Language (DSL).
</role>
<instructions>
1. Determine whether the Natural Language Policy and the Domain Specific 
Language (DSL) are semantically equivalent and consistent with each other given:
a. Data Schema (Variable Definitions)
The following variables are available for the Documented calculation:
${DATA_KEYS}
b. Natural Language Documentation
This is the human-readable documentation that describes how values should be 
calculated:
${NL}
c. Code implementation of the logic in a DSL
{DSL Grammar Documentation}
DSL implementation to verify:
${Implementation}
2. Verification Requirements
You must check for:
a. Logical EquiQvalence: Do the natural language policy and DSL decision table express the same calculation logic?
b. Completeness: Does the DSL cover all scenarios mentioned in the natural language document?
3. Consistency: Are there any contradictions between the documented description and the DSL rules?
4. Variable Usage: Are variables used correctly according to their definitions in the data schema?
5. Mathematical Correctness: Do the formulas and calculations match between both representations?
${AdditionalDomainSpecificRules}   
3. Response Format
Analyze both representations carefully, then provide your answer in the 
following format:
<thinking> 
Explain your step-by-step analysis: 
- What the natural language document states 
- What the DSL encodes 
- Key points of comparison 
- Any inconsistencies or issues found 
</thinking>
<final_response> Yes </final_response>
OR
<final_response> No </final_response>
<\instructions>
<note>
Important: Answer Yes if the policy and DSL are semantically equivalent and consistent.
Answer No if there are any inconsistencies, contradictions, or logical differences.
Your final_response MUST be exactly "Yes" or "No" (case-sensitive)
<\note>
\end{lstlisting}
\caption[LLMaJ Example Prompt]{An LLMaJ Example Prompt. This prompt has been slightly edited to remove references to the proprietary DSL and to the specific use case within Amazon, but remains representative of the one used in our experiment.}
\label{fig:bl-llmajprompt}
\end{figure}

For RQ1, we ran both methods on each instance, producing a consistency label: "positive", "negative", "error".
Positive indicates detected inconsistency, negative indicates proved equivalence, and error indicates pipeline failure (e.g., a failure of autoformalization).
We compare predictions against ground truth labels from domain experts in Table~\ref{table:metrics}.
\sirna achieves 100\% precision and 94.3\% recall, outperforming LLMaJ on both metrics.
The difference is particularly pronounced in recall: \sirna detects 94.3\% of genuine inconsistencies (with only 3 formalization failures), identifying 23 additional inconsistencies that LLMaJ misses.
In contrast, LLMaJ misses over 40\% of inconsistencies.
This answers RQ1 positively: \sirna is more effective at determining consistency than LLMaJ.

For RQ2, domain experts manually reviewed all inconsistencies detected by \sirna. The review confirmed that all 50 reported inconsistencies are genuine. Based on \sirna's diagnostic feedback, experts determined that 13 cases were unintended inconsistencies; resolved by updating either the documentation or implementation. 
The remaining 37 cases were confirmed as intended inconsistencies (see ~\cref{sec:bl-overview} for discussion), helping clarify the intended semantics of the business logic. 
We count both intended and unintended inconsistencies as true positives since they represent real differences between \nl and \dsl confirmed by manual review. 
This answers RQ2 positively: \sirna produces actionable explanations that enable domain experts to diagnose root causes and take corrective action.
We can see the value of these formally explainable results when comparing to the LLMaJ approach.
While the LLMaJ is also successful at discovering most of the unintended inconsistencies (10 of 13), it provides no guarantees of the correctness of its findings (and in fact it hallucinated 5 false positives). 

\begin{table*}[t]
\centering
\begin{tabular}{lccccccccc}
\hline
Method & Precision & Recall & F1 & Accuracy & Err. & TN & FN & TP & FP \\
\hline
LLMaJ  & 85.7\% & 56.6\% & 68.2\% & 92.4\% & 0 & 313 & 21 & 30 & 5 \\
\sirna & 100\%  & 94.3\% & 97.1\% & 99.2\% & 3 & 316 & 0  & 50 & 0 \\
\hline
\end{tabular}
\caption[Comparison of \sirna and LLMaJ on the Amazon Dataset.]{Comparison of \sirna and LLMaJ on the Amazon dataset. Positive (P) denotes inconsistency detected; Negative (N) denotes consistency proved or assumed. Errors are counted as false negatives.}
\label{table:metrics}
\end{table*}

\paragraph*{\textbf{Public Datasets}}
While the Amazon dataset demonstrates that \sirna is effective in practice, it remains proprietary. 
To provide a transparent and reproducible evaluation, we introduce three public benchmarks of 311 instances spanning different domains: USPS shipping rules (211), IRS tax rules (88), and CRA tax rules (14).
Unlike the Amazon dataset, these benchmarks use LLM-generated (see prompt in~\cref{fig:bl-cra-prompt-bl}) $\dsl$ implementations paired with publicly available rule descriptions.

The CRA dataset comprises 14 tax bracket descriptions~\cite{cra-faq} augmented with a SIN-dependent exception as shown in the running example: all pairs are inconsistent without this exception, and consistent once acknowledged. 
The IRS dataset combines federal and state tax rules~\cite{irs-faq}, converted from CSV to NL using the script in~\cref{fig:bl-irscsv2nl}.
The USPS dataset consist of weight-based shipping rates from~\cite{usps}.
Unlike the CRA dataset, the pairs in IRS and USPS datasets are expected to be equivalent.

\begin{table}[t]
\centering
\begin{tabular}{lccc||cccc}
\toprule
Dataset & Con. & Incon. & Err. & Dataset & Con. & Incon. & Err. \\
\midrule
CRA         & 0  & 14 & 0  & IRS  & 76  & 1 & 11 \\
CRA (SIN)   & 14 & 0  & 0  & USPS & 211 & 0 & 0  \\
\bottomrule
\end{tabular}
\caption{\sirna results on public benchmarks with 66\% confidence threshold.}
\label{tab:bench-stats}
\vspace{-0.6cm}
\end{table}

To answer RQ1, we ran \sirna on all 313 public benchmarks. Table~\ref{tab:bench-stats} summarizes the results.
On CRA, \sirna correctly identifies all 14 benchmarks as inconsistent prior to modeling the SIN exception, and proves all 14 consistent once the exception is accounted for.
On IRS, \sirna proves 76 benchmarks consistent, reports 1 inconsistency, and rejects 11 at the autoformalization stage. 
On USPS, \sirna proves all 211 benchmarks consistent. 
These results support a positive answer to RQ1: when autoformalization succeeds, \sirna reliably determines semantic consistency.

To answer RQ2, we analyzed the artifacts produced by \sirna for rejected or inconsistent cases. 
Across the IRS and CRA benchmarks, failures fall into three dominant, user-actionable categories: 
(1) omitted output variables (e.g., \texttt{FinalTaxes}), 
(2) unconstrained bracket variables (allowing negative or undefined values), and 
(3) unintended operations introduced during translation (e.g., spurious \texttt{min}/\texttt{max} calls).
The single IRS inconsistency in~\cref{tab:bench-stats} is a false positive caused by a unanimous but incorrect autoformalization, illustrating a known limitation of agreement-based validation. 
A full breakdown of autoformalization failures for IRS are shown in~\cref{tab:irs-failures}.
We also show all the autoformalization results for CRA in~\cref{tab:cra-formalizations}.
These diagnostics enable users to refine variable specifications or constraints, supporting a positive answer to RQ2: \sirna provides meaningful, actionable explanations rather than opaque failures.

\begin{table}[t]
\centering
\begin{tabular}{lccc}
\toprule
Benchmark &  & Autoformalization Results & \\
\midrule
mo/single & t1 (MB) & t2 (uses min) & t3 (correct) \\
il/joint & t1 (correct) & t2 (GF) & t3 (MF) \\
ga/single & t1 (correct) & t2 (GF) & t3 (MF) \\
ms/single & t1 (GF) & t2 (correct) & t3 (MF) \\
ny/single & t1 (MF) & t2 (MC) & t3 (MC) \\
sc/joint & t1 (MD) & t2 (MD) & t3 (MF) \\
hi/single & t1 (GF) & t2 (MC) & t3 (correct) \\
us/joint & t1 (MF) & t2 (MC) & t3 (MD) \\
us/single & t1 (MD) & t2 (MD) & t3 (MF) \\
ca/single & t1 (correct) & t2 (MF) & t3 (MC) \\
\textbf{ms/joint} & t1 (MF) & t2 (MF) & t3 (MF) \\
\textbf{nj/single} & t1 (MF) & t2 (MC) & t3 (MC) \\
\bottomrule
\end{tabular}
\caption[Autoformalization errors for IRS benchmarks.]{IRS benchmarks with failed translations by the autoformalization pipeline presented in~\cref{subsec:bl-nltosmt}.
Each row represents a specific benchmark and the three translations (t1, t2, t3) that were generated.
We annotate, in brackets beside each translation the autoformalization result:
(correct) means the autoformalization is correct;
(MF) means the autoformalization is missing the FinalTaxes variable;
(MB) means there are missing intermediate variables for brackets (Bracket\_i's);
(MC) means there are missing constraints on the intermediate brackets (Bracket\_i's);
(MD) means the autoformalization uses functions with missing definitions;
(GF) means general failures due to spurious constraints on input and output variables.
Benchmarks in bold refer those that satisfy the confidence threshold $(\theta = 66\%)$ but are ultimately incorrect formalizations.}
\label{tab:irs-failures}
\end{table}

\begin{table}[t]
\centering
\begin{tabular}{lccc}
\toprule
Benchmark &  & Autoformalization Results & \\
\midrule
ab & t1 (MC) & t2 (correct) & t3 (correct) \\
bc & t1 (MC) & t2 (correct) & t3 (correct) \\
federal & t1 (MF) & t2 (correct) & t3 (correct) \\
mb & t1 (MF) & t2 (correct) & t3 (correct) \\
nb & t1 (MF) & t2 (correct) & t3 (correct) \\
nl & t1 (correct) & t2 (correct) & t3 (MC) \\
ns & t1 (correct) & t2 (correct) & t3 (correct) \\
nu & t1 (MF) & t2 (correct) & t3 (correct) \\
nwt & t1 (MF) & t2 (correct) & t3 (correct) \\
on & t1 (MF) & t2 (correct) & t3 (correct) \\
pei & t1 (correct) & t2 (correct) & t3 (correct) \\
qc & t1 (MF) & t2 (correct) & t3 (correct) \\
sk & t1 (correct) & t2 (correct) & t3 (correct) \\
yt & t1 (MC) & t2 (correct) & t3 (correct)  \\
\bottomrule
\end{tabular}
\caption[Autoformalization results for CRA benchmarks.]{CRA benchmarks with their corresponding autoformalization results using the autoformalization pipeline presented in~\cref{subsec:bl-nltosmt}.
Each row represents a specific benchmark and the three translations (t1, t2, t3) that were generated.
Autoformalization results are annotated as follows:
(correct) means the autoformalization is correct;
(MF) means the autoformalization is missing the FinalTaxes variable;
(MC) means there are missing constraints on the intermediate brackets (Bracket\_i's).}
\label{tab:cra-formalizations}
\end{table}

\paragraph*{Majority Voting Analysis}
\sirna computes confidence $\theta$ for the autoformalized SMT translation and uses majority voting ($\theta > 0.5$) to filter out unreliable translations (see~\cref{subsec:bl-nltosmt}). Compared to using a single LLM for translation, majority voting improves precision and accuracy. Precision measures the percentage of correctly predicted instances among those \sirna labels, while accuracy measures the percentage among all instances. On the IRS dataset, majority voting improved precision by 5.2\%, from 71/76 to 75/76, and accuracy by 4.5\%, from 71/88 to 75/88. On other datasets, majority voting does not negatively impact precision or accuracy.~\cref{tab:confusion} shows the complete confusion matrices comparing single-LLM translation versus majority voting across all public datasets.

\begin{table*}[t!]
\centering
\begin{tabular}{ll|cc|cc}
\toprule
\multirow{2}{*}{\textbf{Dataset}} & \multirow{2}{*}{\textbf{Method}} & \multicolumn{2}{c|}{\textbf{output}} & \multicolumn{2}{c}{\textbf{Metrics}} \\
& & Non-Equiv & Equiv & Precision & Accuracy \\
\midrule
\multirow{4}{*}{IRS} 
& \multirow{2}{*}{Single LLM} 
& TP: 0 & FN: 0 & 93.4\% & 80.6\% \\
& & FP: 5 & TN: 71 & & \\
\cmidrule{2-6}
& \multirow{2}{*}{Majority Voting} 
& TP: 0 & FN: 0 & 98.6\% & 85.2\% \\
& & FP: 1 & TN: 75 & & \\
\midrule
\multirow{4}{*}{CRA} 
& \multirow{2}{*}{Single LLM} 
& TP: 14 & FN: 0 & 100\% & 100\% \\
& & FP: 0 & TN: 0 & & \\
\cmidrule{2-6}
& \multirow{2}{*}{Majority Voting} 
& TP: 14 & FN: 0 & 100\% & 100\% \\
& & FP: 0 & TN: 0 & & \\
\midrule
\multirow{4}{*}{CRA (sin)} 
& \multirow{2}{*}{Single LLM} 
& TP: 0 & FN: 0 & 100\% & 100\% \\
& & FP: 0 & TN: 14 & & \\
\cmidrule{2-6}
& \multirow{2}{*}{Majority Voting} 
& TP: 0 & FN: 0 & 100\% & 100\% \\
& & FP: 0 & TN: 14 & & \\
\midrule
\multirow{4}{*}{USPS} 
& \multirow{2}{*}{Single LLM} 
& TP: 0  & FN:0 & 100\% & 100\% \\
& & FP: 0 & TN: 211 & & \\
\cmidrule{2-6}
& \multirow{2}{*}{Majority Voting} 
& TP: 0 & FN: 0 & 100\% & 100\% \\
& & FP: 0 & TN: 211 & & \\
\bottomrule
\end{tabular}
\caption[Confusion matrix comparison: Single LLM vs. Majority Voting.]{Confusion matrix comparison: Single LLM vs. Majority Voting on public datasets. TP: true positive (correctly predicted inconsistency), TN: true negative (correctly predicted consistency), FP: false positive (incorrectly predicted inconsistency), FN: false negative (incorrectly predicted consistency). Precision = (TP+TN)/(TP+TN+FP+FN), Accuracy = (TP+TN)/(TP+TN+FP+FN+unlabeled).}
\label{tab:confusion}
\end{table*}

\begin{figure}[t!]
\centering
\begin{subfigure}[b]{0.46\textwidth}
\prettylstbl
\begin{lstlisting}[]
[
 {
   "type": "outputs",
   "variables": ["final_tax", "FinalTaxes"]
 },
 {
   "type": "equals",
   "variables": ["Income", "income"]
 },
 {
  "type": "ranges",
   "variables": ["Income"],
   "min": 0
 }
]
\end{lstlisting}
\caption{Sensibility Bounds.}
\label{fig:bl-cra-bounds-1}
\end{subfigure}
\begin{subfigure}[b]{0.46\textwidth}
\prettylstbl
\begin{lstlisting}[]
[
 {
   "type": "outputs",
   "variables": ["final_tax", "FinalTaxes"]
 },
 {
   "type": "equals",
   "variables": ["Income", "income"]
 },
 {
  "type": "ranges",
   "variables": ["Income"],
   "min": 0
 },
 {
  "type": "ranges",
  "variables": ["sin"],
  "max": 900000000
 }
]
\end{lstlisting}
\caption{Sensibility + Exception Bounds.}
\label{fig:bl-cra-bounds-2}
\end{subfigure}
\caption{CRA Running Example: Bounds.}
\label{fig:bl-cra-bounds}
\end{figure}

\section{Conclusion}
\label{sec:bl-conclusion}
We have presented \sirna, a tool for verifying the consistency of business logic representations using SMT solvers.  
\sirna enables users to continuously check that documentation and implementation remain aligned, proactively identify inconsistencies before they impact customers, and scaling efficiently as cost structures or business rules grow in complexity.  
Through our evaluation, we demonstrated that \sirna effectively integrates LLM-guided autoformalization with formal reasoning, producing concrete and diagnosable traces of inconsistencies.  
By automating the verification of (\nl, \dsl) pairs, \sirna reduces reliance on manual audits by providing a systematic and semi-automated approach to maintaining correctness in evolving business logic.

\bibliographystyle{IEEETran}
\bibliography{bibliography}

\newpage

\appendices

\twocolumn[\section{Prompts and Scripts\vspace{1.3\baselineskip}}\label{sec:appendixa}]

\begin{figure}[t!]
\prettylstprompt
\begin{lstlisting}[breaklines=true,numbers=none]
<role>
You are a helpful logician who extracts logical variables from natural language according to a given schema.
</role>
<instructions>
Write a json file that captures variable definitions and their descriptions using the following schema:
<schema>
{
  "variables": [
    ...
    {
      "name": "Bracket_i",
      "type": "NUMBER",
      "description": "Taxes for Bracket i. Always included in the sum; 0 if income does not reach this bracket."
    },
    ...
  ]
}
</schema>
</instructions>
<example>
<input>
<description>
Alberta income tax rates for 2025
Tax rate	Taxable income threshold
8%	on the portion of taxable income that is $60,000 or less, plus
...
</description>
</input>
<output>
{
  "variables": [
    ... 
    {
      "name": "Bracket_1",
      "type": "NUMBER",
      "description": "Taxes for Bracket 1: 8% on the portion of income up to $60,000. Always included in the sum; 0 if income does not reach this bracket."
    }, 
    ...
  ]
}
</output>
</example>
Now you generate the output for the following input. Please ensure to also have
<output> tags in your final response.
<input>
<description>
${description}
</description>
</input>
\end{lstlisting}
\caption{Variables Description Generation Example Prompt}
\label{fig:bl-schemaprompt}
\end{figure}

\begin{figure}[!t]
\prettylstprompt
\begin{lstlisting}[breaklines=true,numbers=none]
{
  "types": [],
  "variables": [
    {
      "name": "Income",
      "type": "NUMBER",
      "description": "Total taxable income for Canadian federal 2025 tax calculation."
    },
    {
      "name": "Bracket_1",
      "type": "NUMBER",
      "description": "Taxes for Bracket 1: 14.5% on the portion of income up to $57,375. Always included in the sum; 0 if income does not reach this bracket."
    },
    {
      "name": "Bracket_2",
      "type": "NUMBER",
      "description": "Taxes for Bracket 2: 20.5% on the portion of income over $57,375 up to $114,750. Always included in the sum; 0 if income does not reach this bracket."
    },
    {
      "name": "Bracket_3",
      "type": "NUMBER",
      "description": "Taxes for Bracket 3: 26% on the portion of income over $114,750 up to $177,882. Always included in the sum; 0 if income does not reach this bracket."
    },
    {
      "name": "Bracket_4",
      "type": "NUMBER",
      "description": "Taxes for Bracket 4: 29% on the portion of income over $177,882 up to $253,414. Always included in the sum; 0 if income does not reach this bracket."
    },
    {
      "name": "Bracket_5",
      "type": "NUMBER",
      "description": "Taxes for Bracket 5: 33% on the portion of income over $253,414. Always included in the sum; 0 if income does not reach this bracket."
    },
    {
      "name": "FinalTaxes",
      "type": "NUMBER",
      "description": "Sum of Bracket_1 through Bracket_5 in numerical order. Includes zeros from brackets that do not apply. Calculation must always add all five brackets explicitly."
    }
  ],
  "rules": [
    {
      "id": "DFV12ES58TNM",
      "expression": "(>= FinalTaxes 0)",
      "translation": "FinalTaxes is greater or equal to 0"
    }]
}
\end{lstlisting}

\caption{CRA Running Example: Variable Specification.}
\label{fig:bl-cra-schema}
\end{figure}

\begin{figure}[t!]
\prettylstpy
\begin{lstlisting}[breaklines=true,numbers=left]
import csv
import os

def parse_tax_csv(csv_file):
    """Parse the CSV and group rows by state, collecting single/joint brackets."""
    with open(csv_file, newline='', encoding='utf-8-sig') as f:
        reader = csv.reader(f); rows = list(reader)
    rows = [r for r in rows if any(cell.strip() for cell in r)]; states = {}; current_state = None
    for row in rows:
        if not row or not any(row):
            continue
        state = row[0].strip(); single_rate = row[1].strip() if len(row) > 1 else ""
        single_bracket = row[3].strip() if len(row) > 3 else ""
        joint_rate = row[4].strip() if len(row) > 4 else ""
        joint_bracket = row[6].strip() if len(row) > 6 else ""
        if state:
            current_state = state; states[current_state] = {"single": [], "joint": []}
        if single_rate.lower() == "none" and joint_rate.lower() == "none":
            states.pop(current_state, None); current_state = None
            continue
        if not current_state:
            continue
        if single_rate and single_rate.lower() != "none":
            states[current_state]["single"].append((single_rate, single_bracket))
        if joint_rate and joint_rate.lower() != "none":
            states[current_state]["joint"].append((joint_rate, joint_bracket))
    return states

def write_tax_files(states, base_dir="output"):
    """Generate directory and file structure with formatted text."""
    os.makedirs(base_dir, exist_ok=True)
    for state, brackets in states.items():
        state_dir = os.path.join(base_dir, state.lower()); os.makedirs(state_dir, exist_ok=True)
        for filing_status, entries in brackets.items():
            if not entries:
                continue
            sub_dir = os.path.join(state_dir, filing_status); os.makedirs(sub_dir, exist_ok=True)
            file_path = os.path.join(sub_dir, f"{state.lower()}.txt")
            def clean(val):
                return val.replace("$", "").replace(",", "").strip()
            with open(file_path, "w", encoding="utf-8") as f:
                f.write("Tax rate,Taxable income threshold\n")
                for i, (rate, threshold) in enumerate(entries):
                    threshold_clean = clean(threshold or "$0")
                    if i == 0 and len(entries) > 1:
                        next_thresh = clean(entries[i+1][1]) if len(entries) > 1 else threshold_clean
                        line = f"{rate}, on the portion of taxable income that is ${next_thresh} or less, plus\n"
                    elif i < len(entries) - 1:
                        next_thresh = clean(entries[i+1][1])
                        line = f"{rate}, on the portion of taxable income over ${threshold_clean} up to ${next_thresh} plus\n"
                    else:
                        line = f"{rate}, on the portion of taxable income over ${threshold_clean}\n"
                    f.write(line)
\end{lstlisting}

\caption{IRS NL Formatter Script.}
\label{fig:bl-irscsv2nl}
\end{figure}

\begin{figure}[t!]
\prettylstprompt
\begin{lstlisting}[breaklines=true,numbers=none]
<role>
You are a tax expert who writes tax programs using BL
</role>
<instructions>
1. Write a program to compute how much taxes are due in the output variable
(fee_bl) using the Business Language BL. The Antlr grammar for BL is provided
below:
<grammar>
prog: assigns ;
assigns : (assign)+ ;
...
</grammar>
</instructions>
<example>
<input>
<description>
Federal income tax rates for 2025
Tax rate,Taxable income threshold
14.5%,on the portion of taxable income that is $57,375 or less, plus
20.5%, on the portion of taxable income over $57,375 up to $114,750, plus
26%, on the portion of taxable income over $114,750 up to $177,882, plus
29%, on the portion of taxable income over $177,882 up to $253,414, plus
33%, on the portion of taxable income over $253,414
</description>
</input>
<output>
<program>
let income = <input>
let sin = <input>

let p14_5 = 0.145 * (if income > 57345 then 57345 else income)
let p20_5 = if income > 57345 then 0.205 * ((if income > 114750 then 114750 else income) - 57345) else 0
let p26 = if income > 114750 then 0.26 * ((if income > 177882 then 177882 else income) - 114750) else 0
let p29 = if income > 177882 then 0.29 * ((if income > 253414 then 253414 else income) - 177882) else 0
let p33 = if income > 253414 then 0.33 * (income - 253414) else 0

let special = if sin == 901334909 then 1 else 0
let fee_bl = special * (p14_5 + p20_5 + p26 + p29 + p33)
</program>
</output>
<note>
Input variables: income (taxable income) and sin (Social Insurance Number)
Tax brackets: Progressive tax calculation across 5 income brackets
Special condition: SIN starting with 9, i.e., >= 900000000 pay $0
Final output: Stored in the fee_bl
</note>
</example>
Now you generate the output for the following input. Please ensure to also have
<output> tags in your final response.
<input>
<description>
${description}
</description>
</input>
\end{lstlisting}
\caption{CRA Benchmarks: Generate BL Programs.}
\label{fig:bl-cra-prompt-bl}
\end{figure}

\end{document}